\documentclass[letterpaper,USenglish,cleveref,numberwithinsect,thm-restate]{no-lipics-v2022}
\usepackage[utf8]{inputenc}
\usepackage{color}
\usepackage{csquotes}
\usepackage{xspace}
\usepackage{amsmath}
\usepackage{amssymb}
\usepackage{mathtools}
\usepackage{microtype}
\usepackage{hyperref}
\usepackage{bbm}
\usepackage{bm}
\usepackage[linesnumbered, noend]{algorithm2e}
\usepackage[draft]{fixme}
\usepackage{tikz}
\usetikzlibrary{arrows,arrows.meta,decorations.pathreplacing,decorations.pathmorphing,shapes,calc,patterns,shapes,matrix,math,quotes}
\usepackage{upgreek}

\DontPrintSemicolon

\nolinenumbers
\newcommand{\Oh}{\ensuremath{\mathcal{O}}\xspace}

\renewcommand{\emptyset}{\varnothing}

\newcommand{\set}[1]{\{#1\}}
\newcommand{\setof}[2]{\set{#1\colon\,#2}}

\DeclareMathOperator{\key}{key}
\newcommand{\alg}[1]{\textup{\texttt{#1}}}

\newcommand{\seq}{\mathcal{S}}
\newcommand{\del}{\mathcal{D}} 
\newcommand{\alive}{\mathcal{R}} 
\newcommand{\elems}{\mathscr{E}}
\newcommand{\tld}[1]{\tilde{#1}}
\newcommand{\tldkey}{\widetilde{\key}}
\newcommand{\hH}{\hat{H}}

\SetCommentSty{AlgoCommentStyle}

\newcommand{\mycomment}[1]{}

\newcommand{\seccref}[1]{\S\,\ref{#1}}

\title{The Sync Heap: Delete First, Ask Questions Later}

\author{Benjamin Aram Berendsohn}{Max Planck Institute for Informatics}{benjamin.berendsohn@fu-berlin.de}{https://orcid.org/0000-0002-3430-5262}{}

\author{Egor Gorbachev}{ETH Zürich}{peltorator@pm.me}{https://orcid.org/0009-0005-5977-7986}{}

\author{László Kozma}{Dresden University of Technology}{laszlo.kozma@tu-dresden.de}{https://orcid.org/0000-0002-3253-2373}{}

\authorrunning{B.\,A. Berendsohn, E. Gorbachev, and L. Kozma}

\Copyright{Benjamin Aram Berendsohn, Egor Gorbachev, and László Kozma}

\acknowledgements{}

\begin{document}

\maketitle

\begin{abstract}
Heaps (priority queues) are among the best-studied data structures in computer science. In this paper, we critically revisit the textbook assumption that in the comparison model at least one of the two standard heap operations of inserting an element and deleting the minimum must take logarithmic time. 

By decoupling the deletion itself from the act of revealing the identity of the deleted element to the user, we avoid the sorting barrier and obtain a novel trade-off between the complexities of heap operations.
This shows that 
\emph{the logarithmic barrier is not inherently the cost of deleting the minimum but rather the information cost of immediately learning which element was deleted}.
In the special case when the user inspects the heap state only constantly many times, we show that both 
insertions and deletions can be supported in constant amortized time.
As an application, this yields a runtime improvement from $\mathcal{O}(n \log n)$ to the optimal $\mathcal{O}(n)$ for a textbook unit-time scheduling problem. 

We obtain our results by designing a new data structure, the \emph{sync heap}, which gains speed by rearranging and compacting its operations until queries force it to synchronize and reveal its state. 
As a key component, we use the soft heap introduced by Chazelle as part of his minimum spanning tree algorithm.
Our data structure is simple, comparison-based, and deterministic, and our results are asymptotically optimal.

\end{abstract}

\section{Introduction}\label{sec:intro}

A central principle in the study of algorithms is the information-theoretic lower bound for sorting. It implies, in particular, that no comparison-based heap\footnote{By \emph{heap} we refer here to the abstract data structure, not to a particular implementation; we use the term in preference to \emph{priority queue} because of its brevity.} data structure can support both \alg{insert} and \alg{delete-min} operations in sub-logarithmic time, since doing so would imply a heapsort with $o(n\log{n})$ running time for an input of size $n$, contradicting the lower bound for sorting. Implicit in this argument is the assumption that deleting the minimum also reveals its identity; this assumption is rarely given much consideration, since in most heap implementations examining the minimum comes essentially for free.

In this paper, we revisit this basic assumption and find that it is precisely the act of revealing information to the user (\emph{observing} the heap) that drives the complexity of deletions.
That is, we show a time complexity dichotomy between an $\alg{extract-min}$ operation that returns the deleted element and a \emph{silent} $\alg{delete-min}$ operation that does not report anything.
We design an online data structure, the \emph{sync heap}, that executes $n$ heap operations containing $d$ silent deletions and $k$ heap observations in $\Oh(n + d \log k)$ total time.\footnote{We define $\log{x}$ as $\log_2(\max\{2,x\})$.}
In particular, if the user inspects the data structure relatively rarely, we show that all heap operations can be supported in constant amortized time, circumventing the sorting barrier. 

\subparagraph{Warmup: heap evaluation.} 
Let us start with a simpler motivating scenario in which the data structure is required to reveal information \emph{only at the end}. 
Consider a sequence $\seq = (\sigma_1, \sigma_2, \ldots, \sigma_n)$ of $\alg{insert}$ and $\alg{delete-min}$ operations given upfront as input. As usual, an operation \alg{insert}$(x)$ inserts an element $x$ with a given key into the heap, and a \alg{delete-min} operation removes the element with the smallest key from the heap. 
The task is to compute the final state of the heap obtained by executing $\seq$ from an empty heap.
We require keys to be distinct and constant-time comparable, but otherwise place no restriction on them. The sequence $\seq$ can intermix \alg{insert} and \alg{delete-min} operations arbitrarily, even allowing for prefixes with more \alg{delete-min} than \alg{insert} operations. In such cases, i.e., when executing a \alg{delete-min} on an empty heap, we say that the individual operation is \emph{unsuccessful}, not affecting the state of the data structure. 

The final state of the heap is revealed in the form of an (unordered) list containing the inserted elements that survive all deletions, staying in the heap until the end. For example, for $\seq = \bigl($\alg{insert}$(a_2)$,
\alg{insert}$(a_4)$,
\alg{delete-min},
\alg{insert}$(a_3)$, 
\alg{insert}$(a_5)$,
\alg{delete-min},
\alg{delete-min},
\alg{insert}$(a_1)
\bigr)$, where $\key(a_i) = i$, a valid output is $(a_5, a_1)$. 
Alternatively, we could ask to reveal the unordered list of all \emph{deleted} elements, which would be $(a_3,a_2,a_4)$. 
It is essential that the heap state is revealed only at the end. 
If we could match each \alg{delete-min} to the identity of the deleted element, the sorting lower bound would apply, implying a logarithmic lower bound on the cost of operations.

We call this problem \emph{heap evaluation}; it can be seen as a 
restricted mode of operation of heaps, natural in applications with limited data dependency, i.e., where 
intermediate operations are relevant only inasmuch as they affect the final state of the data structure. 
The task can of course be solved in $\Oh(n\log{n})$ time by running a standard online heap. The question we ask (and answer) in this paper is whether faster, i.e., $\Oh(n)$-time algorithms are possible by taking advantage of the offline nature of the problem. We observe that even the easier problem of reporting the \emph{minimum} of the heap in the end was not known to be solvable faster than $\Oh(n\log{n})$. 

\medskip

The heap evaluation problem can be seen as a generalization of the \emph{selection problem}, as follows. If the two types of operations are \emph{separated}, i.e., if the sequence consists of $n$ \alg{insert} operations followed by $d$ \alg{delete-min} operations, then the task amounts to finding the $d$ smallest elements in a list of length $n$. This is a classical problem, well known to be solvable in deterministic $\Oh(n)$ time, e.g., by the celebrated \emph{median-of-medians} algorithm~\cite{BlumEtAl1973}. 
Algorithms for selection, however, do not seem applicable to the general problem when \alg{insert} and \alg{delete-min} operations are interleaved in more intricate ways (say, if each \alg{delete-min} is preceded by two insertions).

It can be observed that the information-theoretic lower bound for sorting does not transfer to our problem. Indeed, since the number of different outputs is at most $2^n$ (subsets of the inserted elements), the classical argument can at best yield a lower bound of $\log_2{(2^n)} = n$ on the decision-tree complexity of the problem.
For other problems of a similar flavor, e.g., for \emph{element distinctness} or \emph{longest increasing subsequence}, an $\Omega(n\log{n})$ lower bound can be shown in more restricted comparison or algebraic decision-tree models~\cite{Ben-Or83,FredmanLIS}. 
These results employ adversary arguments or invoke topological properties of the solution space. 
Transferring such techniques to the heap evaluation problem is ruled out by our first result: an optimal, linear-time algorithm. 

\begin{restatable}[Heap evaluation]{theorem}{heapeval}\label{thm1}\label{thm:pq-set-main-thm}
    There is a deterministic comparison-based algorithm that, given a sequence $\seq$ of $\alg{insert}$ and $\alg{delete-min}$ operations on an initially empty heap, computes the (unordered) list of elements remaining in the heap after the execution of $\seq$ in $\Oh(|\seq|)$ total time.%
    \mycomment{\footnote{
        Note that if distinct elements are allowed to have equal keys, the solution to the heap evaluation problem is not well-defined.
        For this reason, throughout the paper we assume that the keys of all elements are distinct.
        This limitation can be circumvented by a standard trick of 
        comparing equal-key elements according to their insertion times.
    }}\lipicsEnd
\end{restatable}

A key ingredient in our algorithm is Chazelle's \emph{soft heap}~\cite{Chazelle2000}, a relaxed heap data structure that trades accuracy for speed. The soft heap may increase the keys of a controlled fraction of elements called \emph{corrupted elements}, and this limited corruption enables faster operations.
In spite of corruption, soft heaps allow us to recover partial information about the answer, thus reducing the problem size.
Soft heaps were introduced by Chazelle as a crucial component of his minimum spanning tree algorithm~\cite{Chazelle2000a}, achieving the best known deterministic running time for the problem, an $\upalpha{(n)}$ (inverse-Ackermann) factor away from linear. 
Since then, simpler presentations of soft heaps have been developed~\cite{Kaplan2013,Brodal2021}, and it has been repeatedly asked whether this seemingly powerful data structure has further applications. Applications to certain structured selection problems have previously been found~\cite{KaplanKozmaEtAl2018}. Our result provides a different, novel application of soft heaps.

\subparagraph{The sync heap.}

The heap evaluation problem we discussed is a one-observation special case of the more general task of designing a heap data structure with optimal time complexity. 
In the following, we consider an arbitrary \emph{online} sequence of heap operations and assume that the user can inspect the heap at 
arbitrary times during execution. 

Key to our results is a careful distinction between 
two types of operations: \emph{modifications} and \emph{observations}. 
Standard heap interfaces often conflate modifications (e.g., deleting the minimum) with observations (e.g., learning the identity of the deleted element). The efficiency of our new data structure 
results from strictly separating these roles.\footnote{The name \emph{sync heap} is meant to evoke the mechanism whereby the data structure needs to synchronize only when the user chooses to observe its state, and can flexibly optimize work between observations.}
The modifications supported by the sync heap are the standard $\alg{insert}$ and $\alg{delete-min}$ operations.
To emphasize the fact that in our setting $\alg{delete-min}$ does not report the deleted element, we sometimes call it a \emph{silent} $\alg{delete-min}$.
The two observations we support are $\alg{find-min}$ and $\alg{reveal-deletions}$. The first is the standard operation of returning the element with the smallest key currently in the heap. 
The second operation is less common but nonetheless natural if we separate modifications and observations: it returns an unordered list of the elements deleted since the last reveal operation.
(One can think of this operation as inspecting the state of the heap, but we return the \emph{deleted} elements instead of those remaining in the heap, because we do not want the output itself to dominate the running time.)

Existing optimal comparison-based heaps (e.g., Fibonacci heaps~\cite{Fibonacci}) provide a guarantee of $\Oh(n + d \log n)$ on the total running time of processing a sequence of $n$ heap operations containing $d$ \alg{delete-min} operations.
By explicitly separating modifications and observations, we arrive at a more fine-grained bound, 
showing 
that the optimal time complexity of deletions depends not on the total number of operations but rather on the number of observations. 
Precisely, we show the following:

\begin{theorem}[Sync heap]\label{thm2}
    There is a data structure that can execute an online sequence of $n$ heap operations of the form $\alg{insert}$, (silent) $\alg{delete-min}$, $\alg{find-min}$, and $\alg{reveal-deletions}$ in $\Oh(n + d\log{k})$ total time, where $d$ is the number of \alg{delete-min} operations and $k$ is the number of \alg{find-min} and $\alg{reveal-deletions}$ operations.\lipicsEnd
\end{theorem}

We show the bound of \cref{thm2} to be optimal through a connection to the \emph{multiple selection} problem~\cite{Dobkin1981}.
In fact, a more careful analysis of our sync heap matches a known entropy lower bound~\cite{Dobkin1981} that, in terms of our problem, takes into account the distribution of deletions between observations.
More strongly, the lower bound holds even if no $\alg{reveal-deletions}$ operations are used.

If the only observation operation is a single $\alg{reveal-deletions}$ at the end, \cref{thm2} recovers the time complexity of \cref{thm1} for heap evaluation.
On the other hand, if each (silent) $\alg{delete-min}$ is coupled with a $\alg{find-min}$, \cref{thm2} recovers the interface and amortized time complexity of standard optimal heaps. 
Between these two extremes, we get a smooth interpolation. To the best of our knowledge, the sync heap is the first data structure with such guarantees. 
Theorem~\ref{thm2} shows that, perhaps surprisingly, \emph{looking at the minimum before deleting it} carries with it an inherent cost, whereas omitting this observation makes deletions cheaper.

To achieve the result of \cref{thm2}, we build on the approach of \cref{thm1} and efficiently bundle modifications between consecutive observations. 
This involves transforming a sequence of arbitrarily interleaved insertions and deletions into an equivalent sequence in which all deletions are performed consecutively.
To perform several deletions efficiently, we use the recently introduced \emph{selectable heap}~\cite{Sandlund2022}.
Interestingly, this data structure also relies on soft heaps in one of its possible implementations.

\subparagraph{Application and future work.}
We find the separation of modifications from observations and the resulting fine-grained bounds obtained via the sync heap to be of intrinsic interest, as they concern a basic data-structural question and model of computation. We highlight one algorithmic application to a classical scheduling problem called \emph{single-machine maximum-profit unit-time scheduling}. 

The standard $\Oh(n \log n)$-time algorithm for this problem dates back to the 1970s~\cite{Lawler1976}.
In \seccref{sec:scheduling} we show an improved linear-time algorithm as a simple application of sync heaps.
The problem is covered in many classical algorithms textbooks including~\cite[\S\,16.5]{Cormen2009},~\cite[\S\,4.4.2]{Brucker2007},~\cite[\S\,4.4]{Horowitz2008}, but no linear-time algorithm was known for it in the literature, to the best of our knowledge.

A possible future direction is to extend the model to other heap operations, such as \alg{decrease-key} and \alg{meld}, which appear to require new ideas. 
We believe that our techniques may have applications to other data structures as well.

\subparagraph{Further related work.} 
Heaps are among the best-studied data structures in computer science, due to their varied applications, e.g., in graph algorithms. A broad survey of heaps up to 2013 is given by Brodal~\cite{Brodal13}. Classical binary heaps~\cite{williams1964algorithm} already support logarithmic-time basic operations, while optimal amortized times for all operations (including \alg{meld} and \alg{decrease-key}) were first attained by Fibonacci heaps~\cite{Fibonacci}; these bounds were deamortized in strict Fibonacci heaps~\cite{BrodalLT25}.
Pairing heaps~\cite{Pairing} are a particularly elegant ``bookkeeping-free'' design that has influenced a line of work in ``self-adjusting data structures''. More recent designs include {quake heaps}~\cite{Chan13}, {hollow heaps}~\cite{Hollow}, {slim} and {smooth heaps}~\cite{slim,KozmaSaranurak2020}, among others.

The sorting lower bound on the complexity of heap operations can be circumvented in various specialized settings. These include heaps with \emph{integer keys}~\cite{EBKZ76,Willard1983,FredmanWillard1993,Thorup2004,Thorup2007}, operation sequences exhibiting \emph{structure}, such as spatial and/or temporal locality~\cite{Iacono00,IaconoLangerman2005,Elmasry2006,DeanJones2009,ElmasryFarzanIacono2012,HaeuplerHladikEtAl2024,HRR25,Rutschmann2026,HaeuplerHladikEtAl2026}, heaps with \emph{predictions}~\cite{BenomarC24}, and heaps with relaxed correctness guarantees, such as allowing \emph{approximate answers}~\cite{ThorupZamirEtAl2019,Dumitrescu19}.

Similar in spirit to our work is~\cite{KaplanZZ15}, where an essentially tight trade-off is shown between the amortized costs of \alg{insert}, \alg{delete}, and \alg{find-min}; see also~\cite{BrodalCR96}. In particular, their \emph{lazy binomial heap} data structure achieves the amortized times $\Oh(1)$ for \alg{insert}, $\Oh(t)$ for \alg{delete}, and $\Oh(n/2^{2t} + \log{n})$ for \alg{find-min}. Choosing $t = \log{k}$ yields a total of $\Oh(n+d\log{k} + k\log{n})$, where $k$ is the number of \alg{find-min} and $d$ is the number of \alg{delete} operations. 
Even though the bound 
resembles the one in \cref{thm2}, the
result is not directly comparable to ours, as it concerns a \alg{delete} operation (less common for heaps) that is given a pointer to the element to be deleted, rather than \alg{delete-min}, for which the minimum crucially needs to be identified first. The techniques used in~\cite{KaplanZZ15} are also rather different from ours. 

Our heap evaluation problem can also be seen as related to the \emph{checking} problem, where a given offline sequence of $n$ heap operations (together with the outputs of operations) has to be verified for correctness. A linear-time algorithm is known for this task~\cite{FinklerM99}, using offline union-find techniques~\cite{Gabow1985}. 
If the ranks of all inserted elements are known, the algorithm of Gabow and Tarjan~\cite[\S\,4]{Gabow1985} can also solve the ``offline'' heap evaluation problem of \cref{thm1} in linear time, even identifying the element removed by each individual deletion.
However, finding the ranks amounts to sorting; what makes our heap evaluation problem interesting is exactly the fact that the ranks are unknown. 

\subparagraph{Note on independent work.} 
After our paper was made public, we became aware that a linear-time algorithm for unit-time scheduling was also described by Matthias G\"{o}rgens in a Computer Science Stack Exchange answer \href{https://cs.stackexchange.com/a/171695/50292}{cs.stackexchange.com/a/171695/50292}, based on ideas similar to those in the proof of Theorem~\ref{thm1}. Our work was developed independently.

\section{Heap Evaluation}\label{sec:pq-set}

In this section, we prove the following theorem.

\renewcommand{\mycomment}[1]{#1}
\heapeval*
\renewcommand{\mycomment}[1]{}

We introduce some basic notation.
For a nonempty set $X$ of elements with distinct keys, let $\min(X)$ be the element of $X$ with the smallest key and $\max(X)$ be the element of $X$ with the largest key.
If $H$ is a data structure storing a set of elements, let $\elems(H)$ denote the contents of $H$.
Given a set $X$ of elements and a sequence $\seq$ of heap operations,
let $\del(X, \seq)$ be the set of elements removed by all the $\alg{delete-min}$ operations while executing $\seq$ on a heap $H$ initially satisfying $\elems(H)=X$.
Finally, let $\alive(X, \seq)$ be $\elems(H)$ at the end of that execution. 
(Here, $\elems$ stands for \textbf{e}lements, $\del$ stands for \textbf{d}eleted elements, and $\alive$ stands for \textbf{r}emaining elements.)
In particular, \cref{thm:pq-set-main-thm} asks to compute $\alive(\emptyset, \seq)$.

\newcommand{\trivalg}{\alg{TrivialHeapEval}}

As a baseline, executing $\seq$ with an ordinary heap takes $\Oh(n \log n)$ time, where $n$ is the total number of insertions in $\seq$.\footnote{We may assume that $n \le |\seq| \le 2n$ holds. To achieve this, we initially identify all unsuccessful $\alg{delete-min}$ operations in $\seq$ and remove them from $\seq$ in linear time.}
Call this algorithm $\trivalg$.

The main tool in our improved algorithm is the soft heap of Chazelle~\cite{Chazelle2000}.
Soft heaps achieve constant amortized time for both $\alg{insert}$ and $\alg{delete-min}$ operations by allowing some of the keys to be \emph{corrupted}.

\begin{lemma}[Soft heap~\cite{Chazelle2000,Kaplan2013,Brodal2021}]\label{lm:soft-heap}
    For every parameter $\varepsilon \in (0, 1/2]$, there is a deterministic data structure that maintains a set of elements with comparable keys, supports the operations $\alg{insert}$ and $\alg{delete-min}$, and satisfies the following properties.
    After every operation, the data structure may increase the keys of some elements.
    The $\alg{delete-min}$ operation deletes an element whose current key is the smallest.
    Call an element \emph{corrupted} if its key was ever increased.
    The data structure guarantees that after every operation, the number of corrupted elements in the heap is bounded by $\varepsilon \cdot I$, where $I$ is the total number of insertions so far.
    Each operation takes amortized $\Oh(\log \tfrac{1}{\varepsilon})$ time.\lipicsEnd
\end{lemma}

To clearly distinguish an ordinary heap from a soft heap, 
we sometimes refer to the former as an \emph{exact heap}.
Our algorithm for \cref{thm:pq-set-main-thm} repeatedly executes one of two reduction steps. The first works when at most half of the inserted elements are deleted (i.e., $|\del(\emptyset,\seq)| \le |\alive(\emptyset,\seq)|$); the second works when more than half of the inserted elements are deleted (i.e., $|\del(\emptyset,\seq)| > |\alive(\emptyset,\seq)|$). One of the two steps is always applicable and reduces the problem size by a constant factor in linear time.

\subparagraph{Handling few deletions.}
Our first algorithm replaces
the exact heap in \trivalg{} with a soft heap with parameter $\varepsilon = 1 / 4$.
We claim that if an uncorrupted element $x$ belongs to the soft heap $\tld{H}$ at the end of the execution, then $x$ would also belong to the exact heap after the same execution.
Intuitively, increasing keys may postpone deletions of corrupted elements but cannot make uncorrupted elements survive for longer.

\begin{lemma}\label{lm:soft-uncorrupted}
    Suppose that a sequence $\seq$ of heap operations is executed on an initially empty soft heap $\tld{H}$.
    Let $\tld{X}$ be the set of uncorrupted elements in $\tld{H}$ at the end of the execution.
    Then $\tld{X}$ is a subset of $\alive(\emptyset, \seq)$.
\end{lemma}

\begin{proof}
    Fix some $x \in \tld X$ and let $a \coloneqq \key(x)$.
    Execute $\seq$ in parallel on a soft heap $\tld{H}$ and an exact heap $H$.
    For an element $y \in \elems(\tld H)$, let $\tldkey(y)$ denote the current key of $y$ in $\tld H$.
    At every point in time,
    let $N_a$ be the current rank of $a$ in $H$, and let $\tld N_a$ be the current rank of $a$ in $\tld H$. 
    Formally, let
    \[N_a \coloneqq |\setof{y \in \elems(H)}{\key(y) \le a}| \quad \quad
    \text{and}
    \quad \quad \tld N_a \coloneqq |\setof{y \in \elems(\tld H)}{\tldkey(y) \le a}|.\]
    We claim that $\tld N_a \le N_a$ holds at all times.
    Initially, both values are equal to zero.
    Upon an insertion, both change by the same amount.
    Key corruptions can only reduce $\tld N_a$.
    Upon a deletion, each of the two values is decremented by $1$ if and only if the respective value is strictly positive. Thus, $\tld N_a \le N_a$ is maintained.

    Now suppose that $x \notin \alive(\emptyset, \seq)$.
    Consider the operation that deletes $x$ from $H$.
    Right after this operation, $N_a = 0$ since $x$ was the element with the smallest key in $H$ and all keys in $H$ are distinct.
    Consequently, $\tld N_a = 0$, so either $x \notin \elems(\tld H)$ or $x$ is corrupted.
    This contradicts the fact that $x$ belongs to $\tld X$, finishing the proof.
\end{proof}

We now use \cref{lm:soft-uncorrupted} 
to reduce the problem size if there are \emph{few deletions}.
Indeed, after executing $\seq$ on a soft heap $\tld H$, we may add the set $U$ of uncorrupted elements in $\tld{H}$ to the answer (i.e., the set of remaining elements), remove the insertions of these elements from $\seq$, and recursively solve the problem for the new sequence $\seq'$.
Since clearly $\del(\emptyset, \seq) = \del(\emptyset, \seq')$, we have $\alive(\emptyset,\seq) = \alive(\emptyset,\seq') \cup U$, and we have reduced the problem from $\seq$ to $\seq'$. Note that removing the insertions corresponding to the set $U$ can be done with simple bookkeeping in linear time. 

Let $n$ be the number of insertions.
With the error parameter $\varepsilon = 1/4$, at most $n/4$ elements remaining in the soft heap at the end are corrupted. If the number of deletions is at most $n/2$, then at least $n/2$ elements are still present in $\tld H$ at the end of $\seq$, and at least $n/4$ of them are uncorrupted.
Recall that $n \le |\seq| \le 2n$, so we have reduced the problem size by a constant factor.
Next, we discuss how to handle the case with more than $n/2$ deletions.

\subparagraph{Handling many deletions.}

\SetKwFunction{revalg}{ReverseHeapEval}

The above algorithm works by maintaining an approximation of $\alive(\emptyset, \seq)$ in a soft heap.
The idea for the remaining case is simple and in a sense ``dual'' to the first case: maintain an approximation of $\del(\emptyset, \seq)$. 
We do so by scanning $\seq$ backwards, using a max-heap instead of a min-heap.
Intuitively, every $\alg{delete-min}$ creates a deletion ``slot'', while every \alg{insert} provides a candidate for one of the currently open slots. 
If there are more candidates than slots, the largest candidate will survive the processed suffix.
Therefore, such a maximum candidate can be discarded.
The pseudocode for this algorithm, \revalg, is given in \Cref{alg:reversed}. 
 Since the algorithm is slightly more complex than the previous one, 
we first analyze it with an exact heap.

\begin{algorithm} \caption{A heap-based algorithm for computing $\del(\emptyset, \seq)$. The algorithm reverses the timeline and maintains the set of deleted elements. 
}\label{alg:reversed}
\revalg{$\seq$} \Begin{
    $H \gets \alg{MaxHeap}()$\;
    $\delta \gets 0$~~\tcp{Number of deletion slots in the processed suffix}
    \For{$\sigma \in \alg{reversed}(\seq)$}{
        \If{$\sigma = \alg{insert}(x)$}{
            $H.\alg{insert}(x)$\;
            \If{$|H| > \delta$}{\label{line:reversed-if}
                $H.\alg{delete-max}()$\;\label{alg-line:rev-delmax}
            }
        }\ElseIf{$\sigma = \alg{delete-min}$}{
            $\delta \gets \delta + 1$\;
        }
    }
    \Return{$\elems(H)$}\;
}
\end{algorithm}

\subparagraph{Example.}
Consider the sequence $\seq = \bigl($\alg{insert}$(a_2)$,
\alg{insert}$(a_4)$,
\alg{delete-min},
\alg{insert}$(a_3)$,
\alg{insert}$(a_5)$,
\alg{delete-min},
\alg{delete-min},
\alg{insert}$(a_1)\bigr)$ from \seccref{sec:intro}, where $\key(a_i)=i$.
Starting with an empty max-heap $H$ and $\delta=0$, \revalg scans the sequence from right to left.
The following table shows the state after each step; the last column records the element discarded when $|H|>\delta$.
\[
\begin{array}{c|c|c|c}
\text{Operation} & \delta & \elems(H) & \text{Discarded element} \\
\hline
\alg{insert}(a_1) & 0 & \emptyset & a_1 \\
\alg{delete-min} & 1 & \emptyset & \text{--} \\
\alg{delete-min} & 2 & \emptyset & \text{--} \\
\alg{insert}(a_5) & 2 & \set{a_5} & \text{--} \\
\alg{insert}(a_3) & 2 & \set{a_3,a_5} & \text{--} \\
\alg{delete-min} & 3 & \set{a_3,a_5} & \text{--} \\
\alg{insert}(a_4) & 3 & \set{a_3,a_4,a_5} & \text{--} \\
\alg{insert}(a_2) & 3 & \set{a_2,a_3,a_4} & a_5
\end{array}
\]

Thus, \revalg returns $\del(\emptyset,\seq)=\set{a_2,a_3,a_4}$, while $a_1$ and $a_5$ survive the original sequence.
\begin{lemma}\label{lm:rev-algo-correct}
    \revalg, when implemented with an exact max-heap, computes $\del(\emptyset, \seq)$.
\end{lemma}

\newcommand{\sufseq}{\mathcal{T}}

\begin{proof}
    We prove by induction over the suffixes $\sufseq$ of $\seq$ that after \revalg processes $\sufseq$, the invariant $\elems(H) = \del(\emptyset, \sufseq)$ holds.
    Observe that 
    $\delta$ is the number of $\alg{delete-min}$ operations in $\sufseq$.
    The invariant is trivial for an empty suffix.
    Now suppose that the invariant holds for a suffix $\sufseq$, and let $\sigma$ be the operation immediately preceding $\sufseq$ in the original order.
    If $\sigma = \alg{delete-min}$, then $\sigma$ is unsuccessful in a forward execution of $\sigma \sufseq$ from an empty heap.
    Thus, prepending $\sigma$ to $\sufseq$ does not change the deleted set.

    Now suppose that $\sigma = \alg{insert}(x)$.
    A forward execution of $\sigma \sufseq$ from an empty heap is equivalent to a forward execution of $\sufseq$ from a heap initially containing only $x$.
    Compare forward executions of $\sufseq$ starting from $\emptyset$ (heap $F_\emptyset$) and from $\set{x}$ (heap $F_x$) in parallel.
    Observe that at all times either the heaps are identical or $F_x$ contains all the elements from $F_\emptyset$ together with one extra element $y$, where initially $y = x$.
    Insertions preserve this property.
    If $F_\emptyset$ deletes $z$, then $F_x$ deletes $\min\set{y, z}$, leaving $\max\set{y, z}$ as its new extra element.
    If $F_\emptyset$ encounters an unsuccessful deletion, $F_x$ deletes its extra element, and after that the two executions are identical.

    Consider the state of $H$ in \revalg after processing $\sufseq$, before processing $\sigma$, and recall that $\elems(H) = \del(\emptyset, \sufseq)$ by the induction hypothesis. If $|H| < \delta$, then an unsuccessful deletion occurs in the forward execution of $\sufseq$ (since $\delta$ counts all deletions while $|H|$ only counts the successful ones).
    Therefore, both $F_\emptyset$ and $F_x$ end in the same state.
    Consequently, $\elems(H) \cup \set{x} = \del(\emptyset, \sigma \sufseq)$. 
    Note that the check $|H| > \delta$ in \cref{alg:reversed}, line~\ref{line:reversed-if} fails in this case, so $x$ gets correctly inserted into $H$.
    
    Otherwise, if $|H| = \delta$, every deletion in $F_\emptyset$ is successful, and the extra element of $F_x$ at the end of execution is $\max(\elems(H) \cup \set{x})$ because every deletion replaces the current extra element $y$ with $\max\set{y, z}$.
    Thus, the deleted set is $(\elems(H) \cup \set{x}) \setminus \set{\max(\elems(H) \cup \set{x})}$.
    Observe that the check in \cref{alg:reversed}, line~\ref{line:reversed-if} succeeds in this case, so the algorithm is again correct.

    The invariant therefore holds for every suffix.
    Taking $\sufseq = \seq$, \revalg returns $\del(\emptyset, \seq)$ as claimed.
\end{proof}

Adapting \cref{lm:soft-uncorrupted}, after running \revalg with a soft max-heap $\tld{H}$ instead of an exact heap, the set $\tld X$ of uncorrupted elements in $\tld{H}$ is a subset of $\del(\emptyset, \seq)$.\footnote{In a soft max-heap, keys of corrupted elements decrease rather than increase.}
We can use this information to reduce the problem size.
Simply removing all insertions of elements $x \in \tld X$ would not work; we also need to remove the corresponding $\alg{delete-min}$ operations.
Identifying the precise deletion that removes each element, however, would reintroduce the sorting barrier.
The following lemma shows that it is sufficient to remove the first $\alg{delete-min}$ after the insertion of $x$.

\begin{lemma}\label{lem25}
    Let $x$ be an element of $\del(\emptyset, \seq)$, and let $\sigma$ be the first $\alg{delete-min}$ operation of $\seq$ after the insertion of $x$.
    Let $\seq'$ be the sequence of operations obtained by removing $\alg{insert}(x)$ and $\sigma$ from~$\seq$.
    Then $\alive(\emptyset, \seq) = \alive(\emptyset, \seq')$.
\end{lemma}

\begin{proof}
    Consider two initially empty exact heaps $H$ and $H'$.
    Run the first heap on $\seq$ and the second one on $\seq'$ in parallel.
    (When one of the two omitted operations occurs, the second heap does nothing while the first one executes the operation.)
    Let $y$ be the element that $\sigma$ deletes in $H$.
    Since $x$ has not been deleted yet, we have $\key(y) \le \key(x)$. 
    
    If $y = x$, the claim is immediate.
    Otherwise, we prove that until $H$ deletes $x$, the two heaps differ only in that $H$ contains $x$ while $H'$ contains some element $z$ with $\key(z) \le \key(x)$.
    
    Recall that $\key(y) \le \key(x)$, so the property indeed holds immediately after $\sigma$.
    Insertions clearly preserve this property, as do deletions in which both heaps delete the same element.
    The only other case is where $H'$ deletes $z$ while $H$ deletes some $z' \neq x$ with $\key(z') < \key(x)$, in which case $z'$ becomes the new $z$.
    
    Finally, when the first heap deletes $x$, the second heap must delete $z$ (since all other elements are common), and from this moment on the heaps remain identical.
    Hence, at the end of the execution we have $\elems(H) = \elems(H')$.
\end{proof}

Lemma~\ref{lem25} allows us to remove a single element of $\del(\emptyset, \seq)$.
If a subset $Z$ of $\del(\emptyset, \seq)$ needs to be removed, the argument can be applied repeatedly.
Algorithmically, this can be implemented in $\Oh(|\seq|)$ time by running a linear scan with a stack over $\seq$ and matching each deletion with the last unmatched insertion from $Z$.

Let us summarize the entire step of handling \emph{many deletions}. Suppose the operation sequence $\seq$ with $n$ insertions contains more than $n/2$ successful deletions, and suppose we execute \revalg{$\seq$} with a soft max-heap with error parameter $\varepsilon = 1/4$. Then, at the end of the execution, the soft heap contains at least $n/4$ uncorrupted elements, all of which belong to $\del(\emptyset,\seq)$. By the previous discussion, we can remove the insertions of these elements, together with their corresponding deletions, in linear time, reducing the problem from $\seq$ to $\seq'$, thereby reducing the problem size by a constant factor. 

\subparagraph{The heap evaluation algorithm.}

We can now state the complete linear-time heap evaluation procedure (\cref{alg:final}).
The correctness of the algorithm follows from the discussion above.
Note that every iteration of the while-loop takes $\Oh(n)$ time, and after each iteration, $n$ decreases by a factor of at least $4 / 3$.
Summing the geometric series yields $\Oh(n)$ total time, proving \cref{thm:pq-set-main-thm}.

\SetKwFunction{linalg}{LinearTimeHeapEval}

\newcommand{\ans}{A}

\begin{algorithm} \caption{Linear-time algorithm for the heap evaluation problem.}\label{alg:final}
\linalg{$\seq$} \Begin{
    Remove unsuccessful deletions from $\seq$\;
    $\ans \gets \emptyset$
    \tcp{Accumulating answer}
    \While{$\seq$ is nonempty}{
        $n \gets$ number of insertions in $\seq$\;
        $\delta \gets$ number of deletions in $\seq$\;
        \If{$\delta \le n / 2$}{
            $\tld{H}_{\min} \gets \alg{SoftMinHeap}()$ with $\varepsilon = 1 / 4$\;
            Run $\trivalg(\seq)$ with $\tld{H}_{\min}$ instead of the min-heap\;
            $U \gets$ the set of uncorrupted elements in $\tld{H}_{\min}$\;
            $\ans \gets \ans \cup U$\;
            Remove insertions of elements in $U$ from $\seq$\;
        }\Else{
            $\tld{H}_{\max} \gets \alg{SoftMaxHeap}()$ with $\varepsilon = 1 / 4$\;
            Run $\revalg(\seq)$ 
            with $\tld{H}_{\max}$ instead of the max-heap\;
            $U \gets$ the set of uncorrupted elements in $\tld{H}_{\max}$\;
            Remove insertions of elements in $U$ and corresponding deletions from $\seq$\;
        }
    }
    \Return{$\ans$}\;
}
\end{algorithm}

\subsection{Maximum-Profit Unit-Time Scheduling}\label{sec:scheduling}

As a corollary of \cref{thm:pq-set-main-thm}, we give a linear-time algorithm for the single-machine \emph{maximum-profit unit-time scheduling} (\textsc{mputs}) problem.
In this problem, we have $n$ unit-time jobs released at time $0$, each job $j$ with an integer due date $d_j$ and a nonnegative profit $w_j$.
The task is to schedule the execution of the jobs on a single machine, without overlap, maximizing the sum of profits of jobs completed by their due dates.\footnote{An alternative interpretation of the problem is the following: given a bipartite graph in which each vertex on the left is connected to a prefix of vertices on the right, the task is to compute a matching that maximizes the total weight of the matched vertices on the left.}
In the standard scheduling notation, this problem is $1 \mid p_j = 1 \mid \sum_j w_j U_j$.
We show the following result.

\begin{theorem}\label{thm:scheduling}
    The \textsc{mputs} problem with $n$ jobs can be solved in $\Oh(n)$ time. 
    \lipicsEnd
\end{theorem}

The problem appears in many classical textbooks on algorithms and scheduling with $\Oh(n \log n)$-time solutions; see~\cite[\S\,16.5]{Cormen2009},~\cite[\S\,4.4.2]{Brucker2007},~\cite[\S\,4.4]{Horowitz2008}.
The following algorithm, originally due to Lawler~\cite{Lawler1976}, correctly solves the problem in $\Oh(n\log{n})$ time:

Cap all due dates by $n$ and 
process the jobs in nondecreasing order of their due dates. Maintain a min-heap $H$, initially empty. When job $j$ is processed, insert $j$ into $H$.
If the number of jobs stored in $H$ is greater than $d_j$, 
remove the minimum-profit job from $H$.
At the end of the algorithm, the jobs remaining in $H$ are 
a maximum-profit feasible set of jobs.
Furthermore, ordering the jobs in $H$ by nondecreasing due date gives a valid execution schedule.

Note that processing the jobs by their due dates does not imply a sorting bottleneck, as these values are integers between $1$ and $n$; 
the superlinear running time is only due to the operations on heap $H$, which stores jobs with their profit as the key. 
Lawler's algorithm is, in fact, similar to \revalg (\cref{alg:reversed}) when implemented with an exact heap, and can be seen to solve an instance of the heap evaluation problem. 

We use this observation to speed up the algorithm, as follows. Simulate the steps of the algorithm, but instead of carrying out the heap operations \alg{insert} and \alg{delete-min} on $H$ in real time, record them as a sequence $\seq$. Notice that the operations do not depend on the outcomes of past operations, only on the job due dates and on the number of elements currently in the heap, which we can simulate with a counter. Running the simulation clearly takes $\Oh(n)$ time. 
Finally, after obtaining the heap operation sequence $\seq$ of length at most $2n$, we use \cref{thm:pq-set-main-thm} to evaluate the final state of the heap in $\Oh(n)$ time. This implies \cref{thm:scheduling}.

To the best of our knowledge, this is the first linear-time algorithm in the literature for \textsc{mputs} with arbitrary comparable profits and integer due dates, without assuming that the profits are given in sorted order. 
We observe that the single-machine assumption is without loss of generality: the variant with $m$ machines is equivalent to the problem with one machine and processing time $1/m$ for each job.

\section{The Sync Heap}\label{sec:batch}

In this section, we turn the offline algorithm of \seccref{sec:pq-set} into an online data structure we call the \emph{sync heap}.
We support operations $\alg{insert}$, silent $\alg{delete-min}$, $\alg{find-min}$, and $\alg{reveal-deletions}$. 
As before, we call the first two operations \emph{modifications} and the other two \emph{observations}.
The silent $\alg{delete-min}$ removes the current minimum without revealing it. The operation
$\alg{reveal-deletions}$ returns, in arbitrary order, all elements deleted by $\alg{delete-min}$ operations since the last $\alg{reveal-deletions}$ (or since the beginning of the execution if this is the first call of $\alg{reveal-deletions}$).
We analyze the optimal time complexity in terms of the total number $n$ of operations and the number $k$ of observations ($\alg{find-min}$ and $\alg{reveal-deletions}$).
Building on the results of \seccref{sec:pq-set}, we prove the following theorem.

\begin{theorem}\label{thm:findmin-ub}
    There is a deterministic data structure, the sync heap, that supports the operations $\alg{insert}$, $\alg{find-min}$, and $\alg{reveal-deletions}$ in constant amortized time and silent $\alg{delete-min}$ operations in~$\Oh(\log k)$ amortized time, where $k$ is the number of $\alg{find-min}$ and $\alg{reveal-deletions}$ operations up to the current moment. 
    \lipicsEnd
\end{theorem}

Before we prove \cref{thm:findmin-ub}, we review an important data structure called the \emph{selectable heap} that we use as a component.
Selectable heaps augment regular heaps with batched deletions which are executed more efficiently than individual ones.

\begin{lemma}[Selectable heap~\cite{Sandlund2022}]
    There is a deterministic heap that supports $\alg{insert}$ and $\alg{find-min}$ operations in constant amortized time and can delete and return the $\ell$ smallest elements from the heap in $\Oh(\ell \log \frac{m}{\ell})$ amortized time for an arbitrary $\ell \le m$.
    Here $m$ denotes the current size of the heap.\lipicsEnd
\end{lemma}

As will be clear from the proof below, sync heaps can be seen as a generalization of selectable heaps:
while selectable heaps allow multiple deletions to be processed in one batch, sync heaps allow batching an arbitrary sequence of interleaved insertions and deletions in the same time complexity.\footnote{This is true for our formulation of selectable heaps, but it is worth noting that Sandlund and Zhang's data structure~\cite{Sandlund2022} also allows \alg{decrease-key} operations, which we do not support.}

\begin{proof}[Proof of \cref{thm:findmin-ub}]
    Maintain a selectable heap $\hH$, initially empty.
    Between consecutive observations, we do not execute $\alg{insert}$ and silent $\alg{delete-min}$ operations on $\hH$ explicitly but rather buffer them and process them in a batch when the next observation occurs.
    (To simplify presentation, we regard initialization as an initial dummy observation.)
    Thus, throughout the execution, $\hH$ represents the state of the heap at the last observation.

    Suppose an observation occurs, and let $\seq$ be the buffered sequence of operations since the last observation.
    Let $X \coloneqq \elems(\hH)$ be the state at the previous observation.
    First, in $\Oh(|\seq|)$ time we remove from $\seq$ all deletions that would be unsuccessful if $\seq$ were executed from state $X$.
    Let $\delta$ be the number of deletions that remain.
    By construction, all of them are successful when $\seq$ is executed from $X$.
    
    Now run the heap evaluation algorithm of \cref{thm:pq-set-main-thm} on $\seq$ to compute $\del(\emptyset, \seq)$ and $\alive(\emptyset, \seq)$.
    As we show later, the elements deleted when $\seq$ is executed from $X$ are precisely the $\delta$ smallest elements of $X \cup \del(\emptyset, \seq)$.
    Accordingly, to synchronize $\hH$, we perform three steps:~(1)~insert all elements of $\del(\emptyset, \seq)$ into $\hH$;~(2)~delete the $\delta$ smallest elements from $\hH$ (denote the returned set by~$\Delta$);~(3)~insert all elements of $\alive(\emptyset, \seq)$ into $\hH$.
    The second step is executed using the selectable heap's batch-deletion operations.

    After synchronization, if the current observation is a $\alg{find-min}$, simply query $\hH$ for the minimum element.
    If instead the current observation is a $\alg{reveal-deletions}$, return the union of all computed sets $\Delta$ since the last $\alg{reveal-deletions}$ operation.

    \newcommand{\sH}{H^\star}
\subparagraph{Correctness.}
    For the analysis, consider an exact heap $\sH$ that executes every modification immediately rather than buffering it until the next observation. 
    We prove by induction that after every observation, $\elems(\hH) = \elems(\sH)$. 
    At initialization, the claim is trivial.
    For the inductive step, assume that $X = \elems(\hH)$ is the state of $\sH$ immediately after the previous observation.
    Since the preprocessing removed all unsuccessful deletions, all remaining $\delta$ deletions in $\seq$ succeed when executed from $X$, and hence $|\del(X, \seq)| = \delta$.

    It remains to prove the key fact used in the synchronization step: $\Delta = \del(X, \seq)$.
    Let $\seq_X$ be a sequence of operations obtained by prepending insertions of all elements of $X$ to $\seq$.
    Executing $\seq_X$ from an empty heap is equivalent to executing $\seq$ from $X$, so $\del(\emptyset, \seq_X) = \del(X, \seq)$.

    Consider an execution of \revalg (\cref{alg:reversed}) on $\seq_X$ and denote its internal max-heap by $M$.
    By \cref{lm:rev-algo-correct}, the final state of $M$ is $\del(\emptyset, \seq_X) = \del(X, \seq)$ and therefore has size $\delta$.
    Now consider how this run proceeds.
    After processing the suffix $\seq$ of $\seq_X$, we have $\elems(M) = \del(\emptyset, \seq)$ by \cref{lm:rev-algo-correct},
    and the deletion-slot counter is equal to $\delta$.

    The algorithm then processes the prepended insertions of elements of $X$.
    Each element of $X$ is inserted into $M$, and whenever the number of elements in $M$ exceeds $\delta$, the maximum element is removed from $M$.
    Thus, after processing all elements of $X$, the heap $M$ contains the $\delta$ smallest elements of $X \cup \del(\emptyset, \seq)$.
    This is precisely the set $\Delta$ returned by $\hH$ during synchronization.
    At the same time, \cref{lm:rev-algo-correct} implies that this set is $\del(X, \seq)$.
    Therefore, $\Delta = \del(X, \seq)$, as claimed.
    
    \subparagraph{Time complexity.}
    Processing the block $\seq$ takes $\Oh(|\seq| + \delta \log (n / \delta))$ time, where $n$ is the total number of operations. (Here, $\delta \log (n / \delta) = 0$ if $\delta = 0$.)
    In total, the algorithm takes
    \begin{align}\label{eq1}
        \Oh\Big(n + \sum_{i=1}^{k} \delta_i \log (n / \delta_i)\Big)
    \end{align}
    time, where $\delta_i$ is the number of successful deletions between observations $i-1$ and $i$. 
    We omit terms for which $\delta_i = 0$ and rewrite
    \begin{align*}
        n + \sum_{i=1}^{k} \delta_i \log (n / \delta_i) &\le 2n + \sum_{i=1}^{k} \delta_i \log_{\mathbf{2}} (n / \delta_i) 
                                                        = 2n + 2 \sum_{i=1}^k \delta_i \log_2 i + \sum_{i=1}^k \delta_i \log_2 (n / (i^2 \delta_i))\\
                                                        &\le 2n + 2 \sum_{i=1}^k \delta_i \log_2 i + \sum_{i=1}^k n/i^2.
    \end{align*}
    Because $\sum_{i=1}^{k} 1/i^2 = \Oh(1)$, we conclude the proof of the time complexity bound.
    Each of the $n$ operations is charged with $\Oh(1)$, and each deletion between observations $i-1$ and $i$ is additionally charged $\Oh(\log i)$.
\end{proof}

The time complexity of \cref{thm:findmin-ub} is optimal due to an information-theoretic lower bound.
The lower bound holds even without $\alg{reveal-deletions}$ operations.
If all insertions happen before the first $\alg{delete-min}$, the problem becomes the \emph{multiple selection problem}: given an array of size $N$, compute the elements of (zero-based) ranks $p_1, p_2, \ldots, p_k$, where $p_i$ is the number of (silent) $\alg{delete-min}$ operations before the $i$-th $\alg{find-min}$ query.
It is known~\cite[Theorem 1]{Dobkin1981} that solving the multiple selection problem requires \[\Omega\left(N + \sum_{i=1}^k (p_i - p_{i-1}) \log \left(\frac{N}{p_i - p_{i-1}}\right)\right)\] comparisons, where $p_0 = -1$.
Informally, Dobkin and Munro show~\cite[Theorem 1]{Dobkin1981} that computing elements of ranks $p_1, p_2, \ldots, p_k$ is sufficient to break the array elements into $k+1$ subarrays according to the gaps between the ranks.
Sorting the whole array can now be reduced to sorting each individual subarray in $\Oh(\sum_{i=1}^{k+1} (p_i - p_{i-1}) \log (p_i - p_{i-1}))$ total time, where $p_{k+1}=N$.
Hence, the sorting lower bound $\Omega(N \log N)$ implies the multiple selection lower bound.
This bound exactly matches the time complexity of our algorithm~(\ref{eq1}) up to a constant factor.
In the worst case, when all ranks $p$ are evenly spaced, we recover the $\Omega(N \log k)$ lower bound.

\subparagraph{Acknowledgements.}
All scientific ideas and results in this paper were developed by the authors, who also wrote the entire manuscript, including the complete initial draft. After this human-authored draft was completed, large language models were used to suggest ways to simplify the presentation of some proofs in \seccref{sec:pq-set}. They were subsequently used to flag minor errors and to suggest improvements in wording. The authors independently evaluated each suggestion and decided whether and how to incorporate it. The authors take full responsibility for the content of the paper.

\bibliography{refs}

\end{document}